\documentclass[a4paper]{article}

\usepackage[
	pdftitle={Topological recursion on the Bessel curve},
	pdfauthor={Norman Do and Paul Norbury},
	ocgcolorlinks,
	linkcolor=linkblue,
	citecolor=linkred,
	urlcolor=linkred]
{hyperref}
\usepackage{amsmath}
\usepackage{amssymb}
\usepackage{amsthm}
\usepackage{bm}
\usepackage{booktabs}
\usepackage[hang,flushmargin]{footmisc} 
\usepackage[left=25mm,right=25mm,top=25mm,bottom=25mm]{geometry}
\usepackage[UKenglish]{isodate}
	\cleanlookdateon
\usepackage{mathpazo}
\usepackage{microtype}
\usepackage{sectsty}
	\sectionfont{\large}
	\subsectionfont{\normalsize}
\usepackage{titlesec}
	\titlespacing{\section}{0pt}{12pt}{0pt}
	\titlespacing{\subsection}{0pt}{6pt}{0pt}
\usepackage[nottoc]{tocbibind}
\usepackage{xcolor}
	\definecolor{linkred}{rgb}{0.6,0,0}
	\definecolor{linkblue}{rgb}{0,0,0.6}

\theoremstyle{plain}
	\newtheorem{theorem}{Theorem}
	\newtheorem{proposition}[theorem]{Proposition}
	\newtheorem{corollary}[theorem]{Corollary}
	\newtheorem{remark}[theorem]{Remark}

\newcommand\blfootnote[1]{
	\begingroup
	\renewcommand\thefootnote{}\footnote{#1}
	\addtocounter{footnote}{-1}
	\endgroup
}

\newcommand{\bb}{\boldsymbol{b}}

\newcommand {\dd}{\mathrm{d}}
\newcommand {\h}{\hbar}
\newcommand{\cl}{\mathcal{L}}
\newcommand{\LL}{\boldsymbol{L}}
\newcommand{\modm}{\mathcal{M}}
\newcommand{\mmu}{\boldsymbol{\mu}}

\newcommand {\x}{\widehat{x}}
\newcommand {\y}{\widehat{y}}
\newcommand {\zz}{\bm{z}}
\newcommand{\bz}{\mathbb{Z}}

\begin{document}

{\large \bfseries Topological recursion on the Bessel curve}

{\bfseries Norman Do and Paul Norbury}

The Witten--Kontsevich theorem states that a certain generating function for intersection numbers on the moduli space of stable curves is a tau-function for the KdV integrable hierarchy. This generating function can be recovered via the topological recursion applied to the Airy curve $x=\frac{1}{2}y^2$. In this paper, we consider the topological recursion applied to the irregular spectral curve $xy^2=\frac{1}{2}$, which we call the {\em Bessel curve}. We prove that the associated partition function is also a KdV tau-function, which satisfies Virasoro constraints, a cut-and-join type recursion, and a quantum curve equation. Together, the Airy and Bessel curves govern the local behaviour of all spectral curves with simple branch points.
\blfootnote{
\emph{Date:} \today \\ 
The first author was supported by the Australian Research Council grant DE130100650.}

~

\hrule

\setlength{\parskip}{0pt}
\tableofcontents
\setlength{\parskip}{4pt}

\section{Introduction} \label{sec:introduction}

The topological recursion of Chekhov, Eynard and Orantin takes as input the data of a spectral curve, essentially a Riemann surface $C$ equipped with two meromorphic functions and a bidifferential satisfying some mild assumptions~\cite{CEyHer,EOrInv}. From this information, it produces so-called correlation differentials $\omega_{g,n}$ on $C$ for integers $g \geq 0$ and $n \geq 1$. Although topological recursion was originally inspired by the loop equations in the theory of matrix models, it has over the last decade found widespread applications to various problems across mathematics and physics. For example, it is known to govern the enumeration of maps on surfaces~\cite{ACNPMod, DMaQua, DNoCou, DMSS, DOPS, KZoVir, NorCou}, various flavours of Hurwitz problems~\cite{BHLM, BMa, DDMTop, DLNOrb, EMSLap}, the Gromov--Witten theory of $\mathbb{P}^1$~\cite{DOSSIde,NScGro} and toric Calabi--Yau threefolds~\cite{BKMPRem, EOrCom, FLZEyn}. There are also conjectural relations to quantum invariants of knots~\cite{BEyAll, DFMVol}. Much of the power of the topological recursion lies in its universality --- in other words, its wide applicability across broad classes of problems --- and its ability to reveal commonality among such problems. 

One common feature of the problems governed by topological recursion is that their associated correlation differentials often possess the same local behaviour. In particular, the fact that their spectral curves generically resemble $x = \frac{1}{2} y^2$ locally lifts to a statement concerning the correlation differentials. The invariants $\omega_{g,n}$ of the {\em Airy curve} $x=\frac{1}{2} y^2$ are total derivatives of the following generating functions for intersection numbers of Chern classes of the tautological line bundles $\cl_i$ on the moduli space of stable curves $\overline{\modm}_{g,n}$~\cite{EOrTop}.
\[
K_{g,n}(z_1, \ldots, z_n) = \frac{1}{2^{2g-2+n}}\sum_{|\mathbf{d}| = 3g-3+n} \int_{\overline{\modm}_{g,n}} c_1(\cl_1)^{d_1} \cdots c_1(\cl_n)^{d_n}\prod_{i=1}^n\frac{(2d_i-1)!!}{z_i^{2d_i+1}}
\]

The usual assumption on spectral curves is that the zeroes of $\dd x$ are simple and away from the zeroes of $\dd y$. (Higher order zeros of $\dd x$ can often be handled via the global topological recursion of Bouchard and Eynard~\cite{BEyThi}.) However, an implicit assumption that appears in the literature is that $y$ has no pole at a zero of $\dd x$, in which case we say that the spectral curve is regular. In previous work~\cite{DNoTop}, the authors consider irregular spectral curves, in which poles of $y$ may coincide with a zero of $\dd x$. If the pole has order greater than one, then that particular branch point makes no contribution to the correlation differentials and can be removed from the spectral curve. On the other hand, when the pole is simple, non-trivial correlation differentials arise. We note that irregular spectral curves do arise ``in nature'', for example in matrix models with hard edge behaviour~\cite{AMMBGW} and the enumeration of dessins d'enfant~\cite{DNoTop, KZoVir}.

The previous discussion leads us naturally to consider the \emph{Bessel curve}, defined by the meromorphic functions\footnote{The name {\em Bessel curve} is derived from its quantum curve, which is given by a modified Bessel equation --- see Section~\ref{sec:quantumcurve}.}
\[
x(z) = \frac{1}{2} z^2 \qquad \text{and} \qquad y(z) = \frac{1}{z}.
\]
For $2g-2+n > 0$, the correlation differentials produced by the topological recursion have an expansion
\[
\omega_{g,n}(z_1, \ldots, z_n) = \sum_{\mu_1, \ldots, \mu_n=1}^\infty U_{g,n}(\mu_1, \ldots, \mu_n) \prod_{i=1}^n \frac{\mu_i \, \dd z_i}{z_i^{\mu_i+1}}.
\]
From these expansion coefficients, we define the Bessel partition function
\[
Z(p_1, p_2, \ldots; \h) = \exp \left[ \sum_{g=1}^\infty \sum_{n=1}^\infty \sum_{\mu_1, \ldots, \mu_n=1}^\infty U_{g,n}(\mu_1, \ldots, \mu_n) \frac{\h^{2g-2+n}}{n!} p_{\mu_1} \cdots p_{\mu_n} \right]
\]
and its associated wave function via the so-called principal specialisation
\[
\psi(z, \h) = \left. Z(p_1, p_2, \ldots; \h) \right|_{p_i=z^{-i}}.
\]

The main theme and motivation behind this paper is that statements concerning the Airy curve and its relation to the Kontsevich--Witten KdV tau-function have analogues in the case of the Bessel curve. In particular, topological recursion applied to the Bessel curve is fundamentally related to the Br\'{e}zin--Gross--Witten (BGW) tau-function for the KdV hierarchy. This is to be expected, since an irregular curve represents so-called hard edge behaviour in matrix models --- see for example the Laguerre model in~\cite{CheHar}, which is related to the BGW tau-function again via matrix model techniques in~\cite{MMSUni}. The modest contribution of this paper is a direct proof of the relationship between the BGW tau-function and topological recursion applied to the Bessel curve. We make this connection via deriving Virasoro constraints for the partition function arising from the topological recursion and comparing these to Virasoro constraints for the BGW tau-function \cite{AleCut,AMMBGW, MMSUni}. Alexandrov has recently proven Virasoro constraints, a cut-and-join equation and a quantum curve for for the BGW tau-function using matrix model methods and a beautiful description of the point in the Sato Grassmannian determined by the tau-function~\cite{AleCut}. Once the link between topological recursion and the BGW tau-function is established in Theorem~\ref{thm:taufunction}, further properties of the partition function arising from topological recursion --- Virasoro constraints, a cut-and-join equation and a quantum curve --- are equivalent to those of Alexandrov.  The topological recursion viewpoint helps to explain these properties---the Virasoro constraints are fundamental to topological recursion particularly via Kazarian's treatment \cite{KZoVir}; the cut-and-join equation is essentially another way to express topological recursion---see Theorem~\ref{thm:evolution}; and the quantum curve is expected to be related via a WKB expansion to topological recursion.  Topological recursion helps to explain the properties above, but it does not explain the relationship with KdV.

{\em Acknowledgements.} The authors would like to thank Alexander Alexandrov for numerous discussions.

\section{Topological recursion on the Bessel curve} \label{sec:recursion}

\subsection{Topological recursion}

We briefly recall the construction of the correlation differentials for a rational spectral curve via topological recursion. A statement of the topological recursion in greater generality --- for example, in the case of higher genus spectral curves, locally-defined spectral curves, or spectral curves with non-simple branch points --- can be found elsewhere in the literature~\cite{BEyThi, DOSSIde, EOrInv}.

\begin{itemize}
\item {\bf Input.} A rational spectral curves consists of the data of two meromorphic functions $x$ and $y$ on $\mathbb{CP}^1$. We assume that each zero of $\dd x$ is simple and does not coincide with a zero of $\dd y$. The topological recursion defines symmetric meromorphic multidifferentials $\omega_{g,n}$ on $(\mathbb{CP}^1)^n$ for $g \geq 0$ and $n \geq 1$.\footnote{By a multidifferential on $C^n$, we mean a meromorphic section of the line bundle $\pi_1^*(T^*C) \otimes \pi_2^*(T^*C) \otimes \cdots \otimes \pi_n^*(T^*C)$ on the Cartesian product $C^n$, where $\pi_i: C^n \to C$ denotes projection onto the $i$th factor. We often drop the symbol $\otimes$ when writing multidifferentials.} We refer to these as \emph{correlation differentials}.

\item {\bf Base cases.} The base cases for the topological recursion are given by
\[
\omega_{0,1}(z_1) = -y(z_1) \, \dd x(z_1) \qquad \text{and} \qquad \omega_{0,2}(z_1, z_2) = \frac{\dd z_1 \otimes \dd z_2}{(z_1-z_2)^2}.
\]

\item {\bf Recursion.} The correlation differentials $\omega_{g,n}$ for $2g-2+n>0$ are defined recursively via the following equation.
\[
\omega_{g,n}(z_1, \zz_S) = \sum_{\dd x(\alpha) = 0} \mathop{\text{Res}}_{z=\alpha} K(z_1, z) \Bigg[ \omega_{g-1,n+1}(z, \overline{z}, \zz_S) + \mathop{\sum_{g_1+g_2=g}}_{I \sqcup J = S}^\circ \omega_{g_1,|I|+1}(z, \zz_I) \, \omega_{g_2,|J|+1}(\overline{z}, \zz_J) \Bigg]
\]
Here, we use the notation $S = \{2, 3, \ldots, n\}$ and $\zz_I = \{z_{i_1}, z_{i_2}, \ldots, z_{i_k}\}$ for $I = \{i_1, i_2, \ldots, i_k\}$. The outer summation is over the zeroes of $\dd x$, which we refer to as \emph{branch points}. The function $z \mapsto \overline{z}$ denotes the meromorphic involution defined locally at the branch point $\alpha$ satisfying $x(\overline{z}) = x(z)$ and $\overline{z} \neq z$. The symbol $\circ$ over the inner summation means that we exclude any term that involves $\omega_{0,1}$. Finally, the recursion kernel is given by
\[
K(z_1,z) = -\frac{\int_\infty^z \omega_{0,2}(z_1, \,\cdot\,)}{[y(z)-y(\overline{z})] \, \dd x(z)}.
\] 
\end{itemize}

In previous work, the authors considered the local behaviour of spectral curves and their correlation differentials, and classified branch points into the following three types~\cite{DNoTop}.
\begin{itemize}
\item {\bf Regular.} We say that a branch point is \emph{regular} if $y(z)$ is analytic there. In this case, there is a pole of $\omega_{g,n}$ of order $6g-4+2n$ at the branch point, for $2g-2+n>0$. Note that some of the previous literature on the topological recursion implicitly assumes that the spectral curves under consideration only have regular branch points.
\item {\bf Irregular.} We say that a branch point is {\em irregular} if $y(z)$ has a simple pole there. In this case, there is a pole of $\omega_{g,n}$ of order $2g$ at the branch point, for $2g-2+n>0$.
\item {\bf Removable.} We say that a branch point is {\em removable} if $y(z)$ has a higher order pole there. In this case, the recursion kernel has a zero at the branch point and there is no contribution to the correlation differentials coming from the residue at the branch point.
\end{itemize}

At a regular branch point, a spectral curve locally resembles the Airy curve, which is given by
\[
x(z) = \frac{1}{2} z^2 \qquad \text{and} \qquad y(z) = z.
\]
This property lifts to the fact that the correlation differentials for an arbitrary spectral curve expanded at a regular branch point behave asymptotically like the correlation differentials for the Airy curve~\cite{EOrTop}. Similarly, the correlation differentials for an irregular spectral curve expanded at an irregular branch point behave asymptotically like the correlation differentials for the Bessel curve, which we examine in detail below~\cite{DNoTop}.

\subsection{The Bessel curve}

Define the {\em Bessel curve} to be the rational spectral curve endowed with the meromorphic functions
\[
x(z) = \frac{1}{2}z^2 \qquad \text{and} \qquad y(z) = \frac{1}{z}.
\]
The base cases of the topological recursion are given by
\[
\omega_{0,1}(z_1) = -y(z_1) \, \dd x(z_1) = - \dd z_1 \qquad \text{and} \qquad \omega_{0,2}(z_1, z_2) = \frac{\dd z_1 \, \dd z_2}{(z_1-z_2)^2}.
\]
The spectral curve has only one branch point, which occurs at $z = 0$, and the local involution there is simply $\overline{z} = -z$. Thus, the recursion kernel can be expressed as
\[
K(z_1, z) = - \frac{\int_\infty^z \omega_{0,2}(z_1, \, \cdot \,)}{[y(z) - y(\overline{z})] \, \dd x(z)} = \frac{1}{2} \frac{1}{z-z_1} \frac{\dd z_1}{\dd z}.
\]

For $2g-2+n > 0$ and positive integers $\mu_1, \ldots, \mu_n$, define the number $U_{g,n}(\mu_1, \ldots, \mu_n)$ via the expansion 
\[
\omega_{g,n}(z_1, \ldots, z_n) = \sum_{\mu_1, \ldots, \mu_n=1}^\infty U_{g,n}(\mu_1, \ldots, \mu_n) \prod_{i=1}^n \frac{\mu_i \, \dd z_i}{z_i^{\mu_i+1}}.
\]
Note that such an expansion must exist, since $\omega_{g,n}$ is meromorphic with a pole only at the branch point $z_i = 0$. By convention, we define $U_{0,1}(\mu_1) = 0$ and $U_{0,2}(\mu_1, \mu_2) = 0$.



\begin{proposition} \label{prop:recursion}
For $2g-2+n>0$ and $S = \{2, 3, \ldots, n\}$,
\begin{align}  \label{TR}
\mu_1 \, U_{g,n}(\mu_1, \mmu_S) &= \sum_{k=2}^n (\mu_1+\mu_k-1) \, U_{g,n-1}(\mu_1+\mu_k-1, \mmu_{S\setminus\{k\}}) \\
&+ \frac{1}{2} \mathop{\sum_{\alpha+\beta=\mu_1-1}}_{\alpha, \beta \text{ odd}} \alpha \beta \Bigg[ U_{g-1,n+1}(\alpha, \beta, \mmu_S) + \mathop{\sum_{g_1+g_2=g}}_{I \sqcup J = S} U_{g_1,|I|+1}(\alpha, \mmu_I) \, U_{g_2,|J|+1}(\beta, \mmu_J) \Bigg].\nonumber
\end{align} 
Moreover, all numbers $U_{g,n}(\mu_1, \ldots, \mu_n)$ can be calculated from the base cases $U_{0,1}(\mu_1) = 0$, $U_{0,2}(\mu_1, \mu_2) = 0$ and $U_{1,1}(1) = \frac{1}{8}$.
\end{proposition}

\begin{proof}
Suppose that the numbers $\widetilde{U}_{g,n}(\mu_1, \ldots, \mu_n)$ are defined from the recursion above and the given base cases. It is straightforward to show that these numbers are uniquely defined and that $\widetilde{U}_{g,n}(\mu_1, \ldots, \mu_n) = 0$ unless $\mu_1, \ldots, \mu_n$ are positive odd integers that sum to $2g-2+n$. In particular, $U_{0,n}(\mu_1, \ldots, \mu_n) = 0$ and the generating function
\[
F_{g,n}(z_1, \ldots, z_n) = (-1)^n \sum_{\mu_1, \ldots, \mu_n=1}^\infty \widetilde{U}_{g,n}(\mu_1, \ldots, \mu_n) \prod_{i=1}^n z_i^{-\mu_i}
\]
is a homogeneous polynomial in $\frac{1}{z_1}, \frac{1}{z_2}, \ldots, \frac{1}{z_n}$ that is odd in each variable. The proposition will follow directly from the fact that $\widetilde{\omega}_{g,n} = \omega_{g,n}$ for $2g-2+n>0$, where the $\widetilde{\omega}_{g,n}$ are total derivatives of these generating functions.
\[
\widetilde{\omega}_{g,n}(z_1, \ldots, z_n) = \dd_{z_1} \cdots \dd_{z_n} F_{g,n}(z_1, \ldots, z_n)
\]
It is straightforward to verify that $\widetilde{\omega}_{1,1} = \omega_{1,1}$ and $\widetilde{\omega}_{0,3} = \omega_{0,3}$ by direct computation. We will now proceed to show that $\widetilde{\omega}_{g,n} = \omega_{g,n}$ by induction on $2g-2+n$.

Start by multiplying both sides of the recursion by $z_1^{-\mu_1-1} z_2^{-\mu_2} \cdots z_n^{-\mu_n}$ and sum over all positive integers $\mu_1, \ldots, \mu_n$ to obtain the following.
\begin{align*}
\frac{\partial}{\partial z_1} F_{g,n}(&z_1, \zz_S) = \sum_{k=2}^n \frac{z_k}{z_1^2 - z_k^2} \left[ \frac{\partial}{\partial z_1} F_{g,n-1}(z_1, \zz_{S \setminus \{k\}}) - \frac{\partial}{\partial z_k} F_{g,n-1}(\zz_S) \right] \\
&+ \frac{1}{2} \left[ \frac{\partial^2}{\partial t_1 \partial t_2} F_{g-1,n+1}(t_1, t_2, \zz_S) \right]_{\substack{t_1=z_1 \\ t_2=z_1}} + \frac{1}{2} \mathop{\sum_{g_1+g_2=g}}_{I \sqcup J = S} \left[ \frac{\partial}{\partial z_1} F_{g_1,|I|+1}(z_1, \zz_I) \right] \left[ \frac{\partial}{\partial z_1} F_{g_2,|J|+1}(z_1, \zz_J) \right]
\end{align*}

Now apply $\frac{\partial}{\partial z_2} \cdots \frac{\partial}{\partial z_n}$ to both sides and introduce the notation $W_{g,n}(z_1, \ldots, z_n) = \frac{\partial}{\partial z_1} \cdots \frac{\partial}{\partial z_n} F_{g,n}(z_1, \ldots, z_n)$.
\begin{align*}
W_{g,n}(z_1, \zz_S) &= \sum_{k=2}^n \frac{\partial}{\partial z_k} \frac{z_k}{z_1^2 - z_k^2} \left[ W_{g,n-1}(z_1, \zz_{S \setminus \{k\}}) - W_{g,n-1}(\zz_S) \right] \\
&+ \frac{1}{2} W_{g-1,n+1}(z_1, z_1, \zz_S) + \frac{1}{2} \mathop{\sum_{g_1+g_2=g}}_{I \sqcup J = S} W_{g_1,|I|+1}(z_1, \zz_I) \, W_{g_2,|J|+1}(z_1, \zz_J)
\end{align*}

Note that the fact that $F_{g,n}$ is odd in each variable implies that $\widetilde{\omega}_{g,n}$ is as well. So after multiplying both sides of the previous equation by $\dd z_1 \cdots \dd z_n$, we obtain the following.
\begin{align*}
\widetilde{\omega}_{g,n}(z_1, \zz_S) &= \sum_{k=2}^n \Bigg[ \dd z_k \frac{z_1^2+z_k^2}{(z_1^2-z_k^2)^2} \widetilde{\omega}_{g,n-1}(z_1, \zz_{S \setminus \{k\}}) - \dd z_1 \frac{\partial}{\partial z_k} \frac{z_k}{z_1^2 - z_k^2} \widetilde{\omega}_{g,n-1}(\zz_S) \Bigg] \\
&- \frac{1}{2 \, \dd z_1} \widetilde{\omega}_{g-1,n+1}(z_1, \overline{z}_1, \zz_S) - \frac{1}{2 \, \dd z_1} \mathop{\sum_{g_1+g_2=g}}_{I \sqcup J = S} \widetilde{\omega}_{g_1,|I|+1}(z_1, \zz_I) \, \widetilde{\omega}_{g_2,|J|+1}(\overline{z}_1, \zz_J)
\end{align*}

Now use the fact that a meromorphic 1-form on $\mathbb{CP}^1$ is equal to the sum of its principal parts, which may be stated as
\[
\widetilde{\omega}(z_1) = \sum_\alpha \mathop{\text{Res}}_{z=\alpha} \frac{\dd z_1}{z_1-z} \widetilde{\omega}(z),
\]
where the sum is over the poles of $\widetilde{\omega}(z)$. Applying this to our situation yields the following, where we have removed terms from the right hand side that do not contribute to the residue at $z=0$.
\begin{align*}
\widetilde{\omega}_{g,n}(z_1, \zz_S) = \mathop{\text{Res}}_{z=0} \frac{1}{2} \frac{1}{z-z_1} \frac{\dd z_1}{\dd z} \Bigg[ &-2 \sum_{k=2}^n \dd z \, \dd z_k \frac{z^2+z_k^2}{(z^2-z_k^2)^2} \widetilde{\omega}_{g,n-1}(z, \zz_{S \setminus \{k\}}) \\
&+ \widetilde{\omega}_{g-1,n+1}(z, \overline{z}, \zz_S) + \mathop{\sum_{g_1+g_2=g}}_{I \sqcup J = S} \widetilde{\omega}_{g_1,|I|+1}(z, \zz_I) \, \widetilde{\omega}_{g_2,|J|+1}(\overline{z}, \zz_J) \Bigg]
\end{align*}

We may rewrite this in the following way, using $\omega_{0,2}(z_1, z_2) = \frac{\dd z_1 \, \dd z_2}{(z_1-z_2)^2}$.
\begin{align*}
\widetilde{\omega}_{g,n}(z_1, \zz_S) = \mathop{\text{Res}}_{z=0} \frac{1}{2} \frac{1}{z-z_1} \frac{\dd z_1}{\dd z} \Bigg[ &\sum_{k=2}^n \left( \omega_{0,2}(z, z_k) \widetilde{\omega}_{g,n-1}(\overline{z}, \zz_{S \setminus \{k\}}) + \omega_{0,2}(\overline{z}, z_k) \widetilde{\omega}_{g,n-1}(z, \zz_{S \setminus \{k\}}) \right) \\
&+ \widetilde{\omega}_{g-1,n+1}(z, \overline{z}, \zz_S) + \mathop{\sum_{g_1+g_2=g}}_{I \sqcup J = S} \widetilde{\omega}_{g_1,|I|+1}(z, \zz_I) \, \widetilde{\omega}_{g_2,|J|+1}(\overline{z}, \zz_J) \Bigg]
\end{align*}

By the inductive hypothesis, we may replace each occurrence of $\widetilde{\omega}$ on the right hand side of the equation with the corresponding $\omega$. Furthermore, we may absorb the first summation into the second to obtain the following.
\begin{align*}
\widetilde{\omega}_{g,n}(z_1, \zz_S) = \mathop{\text{Res}}_{z=0} \frac{1}{2} \frac{1}{z-z_1} \frac{\dd z_1}{\dd z} \Bigg[ & \omega_{g-1,n+1}(z, \overline{z}, \zz_S) + \mathop{\sum_{g_1+g_2=g}}_{I \sqcup J = S}^\circ \omega_{g_1,|I|+1}(z, \zz_I) \, \omega_{g_2,|J|+1}(\overline{z}, \zz_J) \Bigg]
\end{align*}
Since this precisely agrees with the topological recursion, we have shown by induction that $\widetilde{\omega}_{g,n} = \omega_{g,n}$ for all $2g-2+n>0$. Hence, $\widetilde{U}_{g,n}(\mu_1, \ldots, \mu_n) = U_{g,n}(\mu_1, \ldots, \mu_n)$ and the proposition follows.
\end{proof}

The correlation differentials produced by the topological recursion satisfy string and dilaton equations, which relate $\omega_{g,n+1}$ and $\omega_{g,n}$~\cite{EOrInv}.

\begin{corollary} \label{cor:stringdilaton}
In the case of the Bessel curve, the string and dilaton equations both reduce to the equation
\[
U_{g,n+1}(1, \mu_1, \ldots, \mu_n) = (2g-2+n) \, U_{g,n}(\mu_1, \ldots, \mu_n).
\]
\end{corollary}

Proposition~\ref{prop:recursion} provides an effective way to calculate the numbers $U_{g,n}(\mu_1, \ldots, \mu_n)$. The only non-zero $U_{g,n}(\mu_1, \ldots, \mu_n)$ in genus up to 4 are given by the following formulas. Observe that the appearance of factorials in each case is due to Corollary~\ref{cor:stringdilaton}.
\begin{align*}
U_{1,n}(1, 1, 1, \ldots, 1) &= \frac{1}{2^3} (n-1)! & U_{4,n}(7, 1, 1, 1, \ldots, 1) &= \frac{175}{2^{19}} (n+5)! \\
U_{2,n}(3, 1, 1, \ldots, 1) &= \frac{3}{2^8} (n+1)! & U_{4,n}(5, 3, 1, 1, \ldots, 1) &= \frac{575}{7 \cdot 2^{19}} (n+5)! \\
U_{3,n}(5, 1, 1, \ldots, 1) &= \frac{15}{2^{13}} (n+3)! & U_{4,n}(3, 3, 3, 1, \ldots, 1) &= \frac{2407}{105 \cdot 2^{18}} (n+5)! \\
U_{3,n}(3, 3, 1, \ldots, 1) &= \frac{21}{5 \cdot 2^{12}} (n+3)!
\end{align*}

\section{Integrability for the Bessel partition function} \label{sec:integrability}

\subsection{Virasoro constraints}

A wide variety of enumerative problems that are governed by the topological recursion have an associated partition function $Z$ that satisfies
\begin{itemize}
\item Virasoro constraints, in the sense that $Z$ is annihilated by a sequence of differential operators that obey the Virasoro commutation relation;
\item an integrable hierarchy, such as the Korteweg--de Vries (KdV), Kadomtsev--Petviashvili (KP), or Toda hierarchies; and
\item an evolution equation of the form $\frac{\partial Z}{\partial s} = MZ$ for some operator $M$ independent of $s$.
\end{itemize}
In particular, this theme has been enunciated by Kazarian and Zograf in the context of enumeration of dessins d'enfant and ribbon graphs~\cite{KZoVir}.

Define the \emph{Bessel partition function}
\[
Z(p_1, p_2, \ldots; \h) = \exp \left[ \sum_{g=1}^\infty \sum_{n=1}^\infty \sum_{\mu_1, \ldots, \mu_n=1}^\infty U_{g,n}(\mu_1, \ldots, \mu_n) \frac{\h^{2g-2+n}}{n!} p_{\mu_1} \cdots p_{\mu_n} \right],
\]
which is an element of $\mathbb{Q}[[\h, p_1, p_2, \ldots]]$. Note that negative powers of $\h$ do not arise in $Z$.
For each non-negative integer $m$, define the differential operator
\begin{equation} \label{eq:virasoro}
L_m = -\frac{m+\frac{1}{2}}{\h} \frac{\partial}{\partial p_{2m+1}} + \sum_{i \text{ odd}} (m+\tfrac{i}{2}) p_{i} \frac{\partial}{\partial p_{2m+i}} + \mathop{\sum_{i+j=2m}}_{i,j \text{ odd}}\frac{ ij}{4} \frac{\partial^2}{\partial p_i \partial p_j} + \frac{1}{16} \delta_{m,0}.
\end{equation}

It is straightforward to verify that the operators $L_0, L_1, L_2, \ldots$ form a representation of half of the Witt algebra, or equivalently, half of the Virasoro algebra with central charge 0. In other words, they obey the Virasoro commutation relations
\[
[L_m, L_n] = (m-n) L_{m+n}, \quad \text{for } m, n \geq 0.
\]

\begin{theorem} \label{thm:virasoro}
For each non-negative integer $m$, we have $L_m Z=0$.
\end{theorem}

\begin{proof}
Write $Z = \exp(F)$ so that $2L_m Z = 0$ is equivalent to
\[
-\frac{2m+1}{\h} \frac{\partial F}{\partial p_{2m+1}} + \sum_{i \text{ odd}} (2m+i) p_i \frac{\partial F}{\partial p_{2m+i}} + \frac{1}{2} \mathop{\sum_{i+j=2m}}_{i,j \text{ odd}} ij \left[ \frac{\partial^2 F}{\partial p_i \partial p_j} + \frac{\partial F}{\partial p_i} \frac{\partial F}{\partial p_j} \right] + \frac{1}{8} \delta_{m,0} = 0.
\]
Extracting the coefficient of $\frac{\h^{2g-3+n}}{n!} p_{\mu_2} \cdots p_{\mu_n}$ from both sides yields the equation
\begin{align*}
(2m+1) \, U_{g,n}(2m+1, \mmu_S) &= \sum_{k=2}^n (2m+\mu_k) \, U_{g,n-1}(2m+\mu_k, \mmu_{S\setminus\{k\}}) \\
&+ \frac{1}{2} \mathop{\sum_{\alpha+\beta=2m}}_{\alpha, \beta \text{ odd}} \alpha \beta \Bigg[ U_{g-1,n+1}(\alpha, \beta, \mmu_S) + \mathop{\sum_{g_1+g_2=g}}_{I \sqcup J = S} U_{g_1,|I|+1}(\alpha, \mmu_I) \, U_{g_2,|J|+1}(\beta, \mmu_J) \Bigg].
\end{align*}
Thus, the fact that $L_m$ annihilates the Bessel partition function $Z$ is equivalent to the recursion of Proposition~\ref{prop:recursion} with $\mu_1 = 2m+1$.
\end{proof}

\subsection{KdV integrability} \label{sec:kdv}

\begin{theorem} \label{thm:taufunction}
The partition function $Z$ is a tau-function for the KdV hierarchy. In particular, $u = F_{p_1p_1}$ satisfies the KdV equation
\[
u_t=u\cdot u_x+\frac{\hbar^2}{12}u_{xxx},\quad u(x,0,0,\cdots)=\frac{\hbar^2}{8(1-x)^2}
\]
for $x=p_1$ and $t=p_3$.  It has trivial dispersionless limit $\displaystyle\lim_{\hbar\to0}u=0$.
\end{theorem}

\begin{proof}
We will show that the Bessel partition function is in fact equal to the Br\'{e}zin--Gross--Witten partition function. Indeed, the Bessel partition function is uniquely defined by the fact that it is annihilated by the Virasoro operators of equation and the normalisation $Z(0) = 1$. On the other hand, the BGW partition function is uniquely defined by the fact that is is annihilated by the Virasoro operators appearing in  \cite{AleCut,AMMBGW, MMSUni}.  Comparing the two sequences of Virasoro operators, we see that they are equal upon setting $t_k = p_k/k$.  Now we simply use the fact that the BGW partition function is a known tau-function for the KdV integrable hierarchy.  The absence of genus zero contributions to $Z$ and leads to the property $\displaystyle\lim_{\hbar\to0}u=0$.
\end{proof}

\subsection{A cut-and-join evolution equation}

The following result shows that the Bessel partition function satisfies an evolution equation. The operator $M$ that appears in the statement resembles the cut-and-join operator for Hurwitz numbers~\cite{GJaTra}.  This operator was also found by Alexandrov \cite{AleCut}.

\begin{theorem} \label{thm:evolution}
The Bessel partition function $Z$ satisfies the equation $\frac{\partial Z}{\partial \h} = MZ$, where
\[
M = \frac{1}{8} p_1 + \frac{1}{2} \sum_{i,j\text{\ odd}} ij p_{i+j+1} \frac{\partial^2}{\partial p_i\partial p_j} + \sum_{i,j\text{\ odd}} (i+j-1) p_ip_j \frac{\partial}{\partial p_{i+j-1}}.
\]
\end{theorem}

\begin{proof}  We give two proofs since one follows methods of Kazarian-Zograf \cite{KZoVir} using Virasoro operators and rather independent of topological recursion, and the other shows that the cut-and-join equation is directly equivalent to topological recursion.

{\em First proof.}
Since the differential operators $L_0, L_1, L_2, \ldots$ of equation~\eqref{eq:virasoro} annihilate the Bessel partition function, so does the following infinite linear combination.
\begin{align*}
\sum_{m=0}^\infty 2p_{2m+1} L_m &= -\sum_{m=0}^\infty p_{2m+1} \frac{2m+1}{\h} \frac{\partial}{\partial p_{2m+1}} + \sum_{m=0}^\infty p_{2m+1} \sum_{i \text{ odd}} (i+2m) p_i \frac{\partial}{\partial p_{i+2m}} \\
&+ \frac{1}{2} \sum_{m=0}^\infty p_{2m+1} \mathop{\sum_{i+j=2m}}_{i,j \text{ odd}} ij \frac{\partial^2}{\partial p_i \partial p_j} + \frac{1}{8} \sum_{m=0}^\infty p_{2m+1} \delta_{m,0} \\
&= -\sum_{m=0}^\infty p_{2m+1} \frac{2m+1}{\h} \frac{\partial}{\partial p_{2m+1}} + M
\end{align*}
Now we simply use the fact that for each monomial appearing in $Z$, the exponent of $\h$ records the weighted degree in $p_1, p_2, \ldots$, where $p_i$ has weight $i$. (This follows from the observation that $U_{g,n}(\mu_1, \ldots, \mu_n)$ is non-zero only when $\mu_1 + \cdots + \mu_n = 2g-2+n$, stated in the proof of Proposition~\ref{prop:recursion}). It follows that
\[
\sum_{m=0}^\infty p_{2m+1} \frac{2m+1}{\h} \frac{\partial Z}{\partial p_{2m+1}} = \frac{\partial Z}{\partial \h}. \qedhere
\]
{\em Second proof.}
For $Z=\exp{F}$ the cut-and-join equation $\frac{\partial Z}{\partial \h} = MZ$ is equivalent to the equation
\[
\frac{\partial F}{\partial\hbar}=\frac{1}{8} p_1F + \frac{1}{2} \sum_{i,j\text{\ odd}} ij p_{i+j+1} \left(\frac{\partial^2}{\partial p_i\partial p_j}F+\frac{\partial}{\partial p_i}F\frac{\partial}{\partial p_j}F\right) + \sum_{i,j\text{\ odd}} (i+j-1) p_ip_j \frac{\partial}{\partial p_{i+j-1}}F
\]
which is equivalent to topological recursion via \eqref{TR}.
\end{proof}

\begin{corollary} \label{th:flow}
The Bessel partition function can be expressed as
\[
Z(p_1, p_2, \ldots; \h) = \exp(\h M) \cdot 1 = \sum_{k=0}^\infty \frac{\h^k}{k!} M^k \cdot 1.
\]
\end{corollary}

This gives an effective way to calculate $Z$. We present here the Bessel partition function $Z$ and corresponding free energy $F = \log (Z)$ up to terms of order $\h^6$.
\begin{align*}
Z(\mathbf{p}; \h) &= 1 + \frac{1}{2^3} p_1 \h + \frac{9}{2^7} p_1^2 \h^2 + \Big( \frac{3}{2^7} p_3 + \frac{51}{2^{10}} p_1^3 \Big) \h^3 + \Big( \frac{75}{2^{10}} p_3 p_1 + \frac{1275}{2^{15}} p_1^4 \Big) \h^4 \\
&+ \Big( \frac{45}{2^{10}} p_5 + \frac{2475}{2^{14}} p_3 p_1^2 + \frac{8415}{2^{18}} p_1^5 \Big) \h^5 + \Big( \frac{1845}{2^{13}} p_5 p_1 + \frac{2025}{2^{15}} p_3^2 + \frac{33825}{2^{17}} p_3 p_1^3 + \frac{115005}{2^{22}} p_1^6 \Big) \h^6 + \cdots \\
F(\mathbf{p}; \h) &= \frac{1}{8} p_1 \h + \frac{1}{16} p_1^2 \h^2 + \Big( \frac{3}{128} p_3 + \frac{1}{24} p_1^3 \Big) \h^3 + \Big( \frac{9}{128} p_3 p_1 + \frac{1}{32} p_1^4 \Big) \h^4 \\
&+ \Big( \frac{45}{1024} p_5 + \frac{9}{64} p_3 p_1^2 + \frac{1}{40} p_1^5 \Big) \h^5 + \Big( \frac{1}{48} p_1^6 + \frac{15}{64} p_3 p_1^3 + \frac{63}{1024} p_3^2 + \frac{225}{1024} p_5 p_1 \Big) \h^6 + \cdots\end{align*}

\begin{remark}
The operator
\[
M = \frac{1}{8} p_1 + \frac{1}{2} \sum_{i,j\text{\ odd}} ij p_{i+j+1} \frac{\partial^2}{\partial p_i\partial p_j} + \sum_{i,j\text{\ odd}} (i+j-1) p_ip_j \frac{\partial}{\partial p_{i+j-1}}
\]
is not an element of the Lie algebra $\widehat{gl(\infty)}$. If it were, then since 1 is a tau-function of the KP hierarchy and the action of $\widehat{GL(\infty)}$ maps KP tau-functions to KP tau-functions, then Corollary~\ref{th:flow} could be used to give another proof that that $Z$ is a KP tau-function. Since $Z$ is a function only of $p_i$ for $i$ odd, one could then deduce that it is a KdV tau-function. One can prove that $M\notin\widehat{gl(\infty)}$ using the fact that $p_1$ is a KP tau-function while $\exp(\hbar M)\cdot p_1$ is not a KP tau-function, which can be observed from the expansion in $\hbar$.
\end{remark}

\section{The quantum curve} \label{sec:quantumcurve}

Consider the wave function $\psi(z, \h)$ formed from the following so-called principal specialisation of the partition function.
\begin{align*}
\psi(z,\h) &= \left. Z(p_1, p_2, \ldots; \h) \right|_{p_i=z^{-i}} \\
&= \exp\left[ \sum_{g=1}^\infty \sum_{n=1}^\infty \sum_{\mu_1, \ldots, \mu_n=1}^\infty U_{g,n}(\mu_1, \ldots, \mu_n) \frac{\h^{2g-2+n}}{n!} z^{-(\mu_1 + \cdots + \mu_n)} \right] \\
&= 1 + \frac{1}{8} \frac{\h}{z} + \frac{9}{128} \frac{\h^2}{z^2} + \frac{75}{1024} \frac{\h^3}{z^3} + \frac{3675}{3268} \frac{\h^4}{z^4} + \frac{59535}{262144} \frac{\h^5}{z^5} + \frac{2401245}{4194304} \frac{\h^6}{z^6} + \frac{57972915}{33554432} \frac{\h^7}{z^7} + \cdots
\end{align*}

\[
\psi_0(z,\h) = \exp(z/\h) z^{-1/2} \psi(z,\h)
\]

\begin{theorem} \label{thm:quantumcurve}
The wave function $\psi(z, \h)$ satisfies the differential equation
\[
\frac{1}{2}z^2\frac{d^2}{d z^2}\psi+\h^{-1} z^2\frac{d}{d z}\psi+\frac{1}{8}\psi=0.
\] 
Equivalently, the modified wave function $\psi_0(z, \h)$ satisfies the differential equation
\[
\left[ \h^2 z^2 \frac{d^2}{d z^2} + \h^2 z \frac{d}{d z} - z^2 \right] \psi_0(z, \h) = 0.
\]
\end{theorem}

\begin{proof}
Start with the evolution equation
\[
\left( \frac{1}{8} p_1 + \frac{1}{2} \sum_{i,j\text{\ odd}} ij p_{i+j+1} \frac{\partial^2}{\partial p_i\partial p_j} + \sum_{i,j\text{\ odd}} (i+j-1) p_ip_j \frac{\partial}{\partial p_{i+j-1}} - \frac{\partial}{\partial \h} \right) Z(p_1, p_2, \ldots; \h) = 0
\]

Consider taking the principal specialisation of this equation.
\begin{align*}
& \left[ \left( \frac{1}{8} p_1 + \frac{1}{2} \sum_{i,j\text{\ odd}} ij p_{i+j+1} \frac{\partial^2}{\partial p_i\partial p_j} + \sum_{i,j\text{\ odd}} (i+j-1) p_ip_j \frac{\partial}{\partial p_{i+j-1}} - \frac{\partial}{\partial \h} \right) p_{\mu_1} \cdots p_{\mu_n} \h^{|\mmu|} \right] \\
&\qquad= \left[ \left( \frac{1}{8} z^{-1} + \frac{1}{2} \sum_{k \neq \ell} \mu_k \mu_\ell z^{-1} + \frac{1}{2} \sum_{k=1}^n (\mu_k^2 + \mu_k) z^{-1} - |\mmu| \h^{-1} \right) p_{\mu_1} \cdots p_{\mu_n} \h^{|\mmu|} \right] \\
&\qquad= \left[ \left( \frac{1}{8} z^{-1} + \frac{1}{2} |\mmu|^2 z^{-1} + \frac{1}{2} |\mmu| z^{-1} - |\mmu| \h^{-1} \right) z^{-|\mmu|} \h^{|\mmu|} \right] \\
&\qquad= \left[ \left( \frac{1}{8} z^{-1} + \frac{1}{2} \frac{d}{d z} z \frac{d}{d z} - \frac{1}{2} \frac{d}{d z} + \h^{-1} z \frac{d}{d z} \right) z^{-|\mmu|} \h^{|\mmu|} \right] \\
&\qquad= \left[ \left( \frac{1}{8} z^{-1} + \frac{1}{2} \frac{d}{d z} + \frac{1}{2} z \frac{d^2}{d z^2} - \frac{1}{2} \frac{d}{d z} + \h^{-1} z \frac{d}{d z} \right) z^{-|\mmu|} \h^{|\mmu|} \right] \\
&\qquad= \left[ \left( \frac{1}{8} z^{-1} + \frac{1}{2} z \frac{d^2}{d z^2} + \h^{-1} z \frac{d}{d z} \right) z^{-|\mmu|} \h^{|\mmu|} \right] \\
\end{align*}

It follows that
\[
\left[ \frac{1}{2} z^2 \frac{d^2}{d z^2} + \h^{-1} z^2 \frac{d}{d z} + \frac{1}{8} \right] \psi(z, h) = 0.
\]

\end{proof}

\begin{corollary}
\[
\psi(z, \h) = \sum_{d=0}^\infty \frac{(2d)!^2}{32^d d!^3} \frac{\h^d}{z^d} = \sum_{d=0}^{\infty} \frac{(2d-1)!!^2}{8^d d!} \left(\frac{\h}{z}\right)^d
\]
\end{corollary}

\begin{proof}
Put $\psi(z, \h) = \sum_{d=0}^\infty a_d \frac{\h^d}{z^d}$, so we have
\begin{align*}
0=\left[ \frac{1}{2} z^2 \frac{d^2}{d z^2} + \h^{-1} z^2 \frac{d}{d z} + \frac{1}{8} \right] \sum_{d=0}^\infty a_d \frac{\h^d}{z^d} &= \sum_{d=0}^\infty \left[ \frac{1}{2} d(d+1) a_d \frac{\h^d}{z^d} - d a_d \frac{\h^{d-1}}{z^{d-1}} + \frac{1}{8} a_d \frac{\h^d}{z^d} \right]\\
&=
\sum_{d=0}^\infty \left[ \frac{1}{2} d(d+1) a_d - (d+1) a_{d+1} + \frac{1}{8} a_d \right] \frac{\h^d}{z^d} 
\end{align*}
Thus $\displaystyle\frac{1}{2} d(d+1) a_d - (d+1) a_{d+1} + \frac{1}{8} a_d = 0 \quad \Rightarrow \quad a_{d+1} = \frac{\frac{1}{2} d(d+1) + \frac{1}{8}}{d+1} a_d = \frac{1}{8} \frac{(2d+1)^2}{d+1} a_d$,

hence
$\displaystyle a_d = \prod_{i=1}^d \frac{a_i}{a_{i-1}} = \prod_{i=1}^d \frac{1}{8} \frac{(2i-1)^2}{i} = \frac{1}{8^d} \frac{(2d-1)!!^2}{d!}.$
\end{proof}

From the viewpoint of Gukov--Su{\l}kowski, we should define the wave function thus:
\[
\psi_0(z, \h) = \exp(z/\h) z^{-1/2} \psi(z, \h).
\] 
Then from the differential equation above, we obtain
\[
\left[ \h^2 z^2 \frac{d^2}{d z^2} + \h^2 z \frac{d}{d z} - z^2 \right] \psi_0(z, \h) = 0.
\]

In terms of the operators $\widehat{x} = x = \frac{1}{2} z^2$ and $\widehat{y} = \h \frac{d}{d x} = \frac{\h}{z} \frac{d}{d z}$, we can write this as
\[
\left[ 2\hat{y}\hat{x}\hat{y} - 1 \right] \psi_0(z, \h) = 0.
\]
This is the quantum curve equation and its semi-classical limit is $2xy^2 - 1 = 0$, from which we recover the spectral curve. This quantum curve was also obtained by Alexandrov in~\cite{AleCut}, and by Bouchard and Eynard~\cite{BEyRec} where its relation to topological recursion was proven as a consequence of a much more general theorem for a large class of rational spectral curves.

Note that the differential equation above is the modified Bessel's equation (after setting $h=1$) with parameter 0. Hence, $\psi_0(z,\h)=K_0(z/\h)$ is the modified Bessel function. The name {\em Bessel curve} derives from this.

\begin{equation} \label{wavef}
\psi_0(z,\h)=\exp(\h^{-1}S_0(z)+S_1(z)+\h S_2(z)+\h^2 S_3(z)+...)
\end{equation}
For $k>1$, this is conjecturally given by
\begin{equation} \label{exactSk}
S_k(z)=\sum_{2g-1+n=k}\frac{1}{n!}\int^z_{\infty}\int^z_{\infty}...\int^z_{\infty}\omega^g_n(z_1,...,z_n)
\end{equation}
where $\omega^g_n(z_1,...,z_n)$ are multidifferentials for each $g\geq 0$, $n>0$ recursively defined on the curve $x=\frac{1}{2}z^2$, $y=1/z$ via topological recursion. This conjecture is addressed by Gukov and Su\l kowski in~\cite{GSuApo} together with the related issue of constructing $\widehat{P}(\x,\y)$ algorithmically from the wave function.

\[
\int^z_{\infty}\int^z_{\infty}...\int^z_{\infty}\omega^g_n(z_1,...,z_n)=\sum_{\mu} (-1)^nU_g(\mu_1, \ldots, \mu_n) \prod_{i=1}^n \left.z_i^{-\mu_i}\right|_{z_i=z}=\sum_{\mu} (-1)^nU_g(\mu_1, \ldots, \mu_n) z^{2-2g-n}
\]
since $\sum\mu_i=2g-2+n$. Hence
\[
\sum_{k>1}\h^{k-1}S_k(z)=\sum_{g,n,\mu}\frac{(-1)^n}{n!}U_g(\mu_1, \ldots, \mu_n) \left(\frac{\h}{z}\right)^{2g-2+n}=\log Z|_{\{p_i=z^{-i},s=-\h\}}
\]
Now $\frac{d}{dx}S_0(z)=-y$ hence $S_0(z)=-z$ and $S_1(z)=-\frac{1}{2}\log\frac{dx}{dz}=-\frac{1}{2}\log z$ so
\[
\psi_0(z,\h)=\exp(\h^{-1}S_0(z)+S_1(z))\cdot Z|_{\{p_i=z^{-i},s=-\h\}}=e^{-z/\h}z^{-\frac{1}{2}}\cdot Z|_{\{p_i=z^{-i},s=-\h\}}
\]
and this gives agrees with the asymptotic expansion of the modified Bessel function $\psi_0=K_0$.

\textsc{School of Mathematical Sciences, Monash University, VIC 3800, Australia} \\
\emph{Email:} \href{mailto:norm.do@monash.edu}{norm.do@monash.edu}

\textsc{School of Mathematics and Statistics, The University of Melbourne, VIC 3010, Australia} \\
\emph{Email:} \href{mailto:pnorbury@ms.unimelb.edu.au}{pnorbury@ms.unimelb.edu.au}

\end{document}